\documentclass[11pt]{article}
\usepackage{fullpage}
\usepackage{amsmath}
\usepackage{amsthm}
\usepackage{authblk} %for \align
\usepackage{amssymb, algorithm, algpseudocode,eqparbox,thmtools,thm-restate,enumerate,verbatim,bm,array,tabu}
\usepackage[usenames,dvipsnames,svgnames,table]{xcolor}
\definecolor{darkgreen}{rgb}{0.0,0,0.9}
\usepackage[colorlinks=true,pdfpagemode=UseNone,
citecolor=OliveGreen,linkcolor=BrickRed,urlcolor=BrickRed,
pdfstartview=FitH,pagebackref=true]{hyperref}

\renewcommand{\P}[1]{{\mathbb{P}}\left[#1\right]}

\newcommand{\E}[1]{{\mathbb{E}}\left[#1\right]}

\newcommand{\st}{\mbox{\rm s.t. }}

\def\equationautorefname~#1\null{%
  equation~(#1)\null
}

\declaretheorem[numberwithin=section]{theorem}
\declaretheorem[sibling=theorem]{lemma}
\declaretheorem[sibling=theorem]{proposition}

\declaretheorem[sibling=theorem]{corollary}

\declaretheorem[sibling=theorem]{fact}

\def\setminus{-}

\newenvironment{proofof}[1]{{\medbreak\noindent \em Proof of #1.  }}{\hfill\qed\medbreak}

\def\eps{{\epsilon}}
\def\R{\mathbb{R}}
\def\C{\mathbb{C}}

\def\bx{{\bf x}}

\def\SS{{\mathcal S}}

\def\by{{{\bf{y}}}}
\def\bx{{{\bf{x}}}}
\def\bz{{{\bf{z}}}}
\def\blambda{{{\bf{\lambda}}}}

\DeclareMathOperator{\image}{Im}

\DeclareMathOperator{\OPT}{OPT}
\begin{document}

\title{Nash Social Welfare, Matrix Permanent, and Stable Polynomials}

\author[1]{Nima Anari}
\author[2]{Shayan Oveis Gharan}
\author[1]{Amin Saberi}
\author[3]{Mohit Singh}

\affil[1]{\small Stanford University.}
\affil[2]{\small University of Washington.}
\affil[3]{\small Microsoft Research, Redmond.}

\date{}
\maketitle

%In this note we extend the result of Marcus, Spielman and Srivastava to Random Spanning tree distributions.
%Given a graph $G=(V,E)$, and a $k$-edge-connected subgraph of edges $F\subseteq E$.
%Suppose we have assigned a vector
\begin{abstract}
We study the problem of allocating $m$ items to $n$ agents subject to maximizing the Nash social welfare (NSW) objective. We write a novel convex programming relaxation for this problem, and we show that a simple randomized rounding algorithm gives a $1/e$ approximation factor of the objective.

Our main technical contribution is an extension of  Gurvits's lower bound on the coefficient of the square-free monomial of a degree $m$-homogeneous stable polynomial on $m$ variables to all homogeneous polynomials. We  use this extension to analyze the expected welfare of the allocation returned by our randomized rounding algorithm.
\end{abstract}

\section{Introduction}
We study the problem of allocating a set of indivisible items to agents subject to maximizing the Nash social welfare (NSW).
We are given a set of $m$ indivisible items and we want to assign them to $n$ agents. An allocation vector is a vector  $\bx\in\{0,1\}^{n\times m}$ such that for each $j$, exactly one $x_{i,j}$ is 1. We assume that agents have additive valuations. That is, each agent $i$ has non-negative value $v_{i,j}$ for an item $j$ and the value that $i$ receives  for an allocation $\bx$ is
$$ v_i(\bx)=\sum_{j=1}^m x_{i,j}v_{i,j}.$$
The NSW objective is to compute an allocation $\bx$ that maximizes the geometric mean of agents' values,
$$ \left(\prod_{i=1}^n v_i(\bx)\right)^{\frac{1}{n}}.$$
%Note that conventionally NSW objective is defined as  the geometric mean of the agents' values. Here, we use the product to simplify our notation.

The above objective naturally encapsulates both fairness and efficiency and has been extensively studied as a notion of fair division (see ~\cite{Moulin,CKMPSW16} and references therein).

Recently, there have been a number of results that study the computational complexity of the Nash social welfare objective.
For additive valuations it is shown that it is NP-hard to approximate the NSW objective within $(1-c)$  \cite{NNRR14,Lee15}, for some constant $c>0$. On the positive side, Nguyen and Rothe {\cite{NR14} designed a $\left(\frac{1}{m-n+1}\right)$ approximation algorithm and  Cole and Gkatzelis \cite{CG15} gave the first constant factor, $\left(\frac{1}{2e^{1/e}}\right)$-approximation. Recently, Cole et al. ~\cite{CDGJMVY16} gave a $\frac12$-approximation for the problem.

A closely related problem, that captures only fairness, is the Santa-Clause problem where the goal is to find an allocation to maximize the minimum value among all agents, i.e., $\max_{\bx}\min_i v_i(\bx)$ which has also been studied recently~\cite{AFS08,AS07,BS10,CCK09}.
% We refer interested readers to \cite{CG15} for a list of applications of this problem.

\subsection{Our Contributions}
Our main contribution is to show an intricate connection between the Nash welfare maximization problem, the theory of real stable polynomials, and the problem of approximating the permanent. We establish this connection in the following manner. We first give a new mathematical programming relaxation for the problem; indeed the standard relaxation has arbitrarily large integrality gap as shown by Cole and Gkatzelis~\cite{CG15}. Our relaxation is a polynomial optimization problem\footnote{It falls in the broad class of geometric programs, where the mathematical program is convex in logarithms of the variables and not the variables itself.} which, despite not being convex in the standard form, can be solved efficiently by a change of variables. We remark that a similar geometric program was used in the context of maximum sub-determinant problem~\cite{NS16}.

More precisely, we study a real stable polynomial $p(y_1,\dots,y_m)$.
We give a simple randomized rounding algorithm such that the expected Nash welfare  of the allocation returned by the algorithm exactly equals the sum of \emph{square-free} coefficients of $p(\by)$.
Thus, our  program needs to maximize  the sum of square-free coefficients of $p(\by)$.
Unfortunately, such an optimization problem in not convex. Instead, we maximize the following proxy
$$ \inf_{\substack{\by>0:\\ \prod_{i\in S} y_i \geq1, \forall S\in {[m]\choose n}}} p(\by).$$
%Note that the logarithm of the above program is convex in $\log y$.

The main part of our analysis is to relate the sum of square-free coefficients of $p(\by)$ to the above proxy.
This desired inequality is a generalization of an elegant result of Gurvits~\cite{Gur06} relating the problem of approximating the permanent of a matrix with the theory of real stable polynomials. We prove this generalization in \autoref{thm:gurvitsextension}. The connection to permanents allows us to use algorithmic results for approximating the permanent due to Jerrum, Sinclair and Vigoda~\cite{JSV04} and we obtain the following result.
\begin{theorem}\label{thm:nashwelfaremain}
There is a randomized polynomial time algorithm for the Nash welfare maximization problem that, with high probability, returns a solution with objective at least $1/e$ fraction of the optimum.
\end{theorem}

We emphasize that unlike the constant factor approximation algorithm by Cole and Gkatzelis \cite{CG15}, our algorithm and its analysis are purely algebraic and completely oblivious to the structure of the underlying market. In particular, unlike other approaches we are not taking advantage of the combinatorial structure of ``spending restricted assignments'' in our rounding algorithms (see \cite{CG15} for more information). This generality makes our approach potentially applicable to a variety of resource allocation problems.

%In addition, we write a novel convex relaxation for the problem.
%As observed by Cole and Gkatzelis, the natural convex relaxation of the problem can have an arbitrary large integrality gap.
%Therefore, Cole and Gkatzelis in their paper use a primal dual framework to obtain their approximation factor.??

The crucial ingredient of our analysis is the following general inequality about real stable polynomials that generalizes the result of Gurvits \cite{Gur06} (see \autoref{thm:gurvits}) that provided an elegant proof of the Van-der-Waerden conjecture.

\begin{theorem}\label{thm:gurvitsextension}
Let $p$ be a degree $n$-homogeneous real stable polynomial in $y_1,\dots,y_m$ with non-negative coefficients. For any set $S\subseteq [m]$, let $c_S$ denote the coefficient of monomial $y^S:=\prod_{i\in S} y_i$. If  $\sum_{S\in{[m]\choose n}} c_S>0$, then
%\begin{eqnarray}\label{eq:gurvitsextension} \left(\sum_{S\in{m\choose k}} \prod_{i\in S} \partial_{y_i}\right) p(y_1,\dots,y_m) &\geq& \frac{m! \cdot (m-k)^{m-k}}{m^m\cdot (m-k)!} \min_{y:  y^S\geq 1, \forall S\in{m\choose k}} p(y_1,\dots,y_m)	\\
\begin{eqnarray}\label{eq:gurvitsextension} \sum_{S\in{[m]\choose n}} c_S &\geq& \frac{m! \cdot (m-n)^{m-n}}{m^m\cdot (m-n)!} \inf_{\substack{\by>0:\\  y^S\geq 1, \forall S\in{[m]\choose n}}} p(\by)	\\
&\geq & e^{-n} \inf_{\substack{\by>0:\\  y^S\geq 1, \forall S\in{[m]\choose n}}} p(\by)	.\nonumber
\end{eqnarray}
\end{theorem}
Note that second inequality follows by \autoref{lem:etomk}, $\frac{m!}{m^m}\cdot \frac{(m-k)^{m-k}}{(m-k)!} \geq e^{-k}$.
By setting $n=m$ in the above statement, we obtain the result of Gurvits as described in \autoref{thm:gurvits}.

It is not hard to see that the above inequality is (almost) tight. For the stable $n$-homogeneous polynomial $p(y_1,\dots,y_m)=(y_1+\dots+y_n)^n$, the LHS is $n!$ and the RHS is $(n/e)^n$.

\section{Preliminaries}
%We write $\partial_{y_i}$ to denote the operator that performs partial differentiation  in $y_i$, i.e., $\partial/\partial y_i$.
For a vector $\by$, we write $\by\leq 1$ to denote that all coordinates of $\by$ are at most $1$.
For an integer $n\geq 1$ we use $[n]$ to denote the set of numbers $\{1,2,\dots,n\}$. For any $m,n$, we let $\binom{[m]}{n}$ denote the collection of subsets of $[m]$ of size $n$.
\subsection{Stable Polynomials}
Stable polynomials are natural multivariate generalizations
of real-rooted univariate polynomials. For a complex number $z$, let
$\image(z)$ denote the imaginary part of $z$.
We say a polynomial $p(z_1,\dots,z_m)\in\C[z_1,\dots,z_m]$ is {\em stable}
if whenever $\image(z_i)>0$ for all $1\leq i\leq m$, $p(z_1,\dots,z_m)\neq 0$. We say $p(.)$ is real stable, if it is stable and all of its coefficients are real. It is easy to see that any univariate polynomial is real stable  if and only if it is real rooted.

We say a polynomial $p(z_1,\dots,z_m)$ is degree $n$-homogeneous, or $n$-homogenous, if every monomial of $p$ has degree exactly $n$. Equivalently, $p$ is $n$-homogeneous if for all $a\in\R$, we have
$$ p(a\cdot z_1,\dots,a\cdot z_m)=a^n p(z_1,\dots,z_m).$$

We say a monomial $z_1^{\alpha_1}\dots z_m^{\alpha_m}$ is {\em square-free} if $\alpha_1,\dots,\alpha_m\in\{0,1\}$. For a set $S\subset 2^{[m]}$ we write
$$ z^S=\prod_{i\in S} z_i.$$
Thus, we can represent a square-free monomial with the set of indices of variables of that monomial.

\begin{fact}\label{fact1}
If $p(z_1,\dots,z_m)$ and $q(z_1,\dots,z_m)$ are stable then $p\cdot q$ is stable.	
\end{fact}
\begin{fact}\label{fact2}
The polynomial $\sum_i a_i z_i$ is stable if $a_i\geq 0$ for all $i$.	
\end{fact}

Polynomial optimization problems involving real stable polynomials with nonnegative coefficients can often be turned into concave/convex programs. Such polynomials are log-concave in the positive orthant:
\begin{theorem}[\cite{Gul97}]
\label{thm:guler}
For any $n$-homogeneous  stable polynomial  $p(x_1,\dots,x_n)$ with nonnegative coefficients, $\ln p(\bx)$ is concave in the positive orthant, $\R_{++}^n$.
\end{theorem}
It is also an immediate corollary of H\"older's inequality that a polynomial with nonnegative coefficients is log-convex in terms of the log of its variables (for more details on log-convex functions see \cite{BV06}).
\begin{fact}
\label{fact:logconvex}
For any polynomial $p(y_1,\dots, y_m)$ with nonnegative coefficients, $\ln p(\by)$ is convex in terms of $\ln \by$. In other words $\ln p(e^{z_1},\dots,e^{z_m})$ is convex in terms of $\bz$.
\end{fact}

%\begin{proof}
%Firstly, note that since $p$ has nonnegative coefficients, $\log p(y)$ is well-defined for any $y\in \R_{++}^m$.
%It is enough to show that the function is concave along any interval in the positive orthant. Let $a\in\R_{++}^m$ and consider the ling $a+tb$ where for any $t\in [0,1]$, $a+tb\in\R_{++}^m$. We show that $\log p(a+tb)$ is concave.
%$$ p(a+tb)=p(t(a/t+b)) = t^k p(a/t+b)$$
%Since $a\in\R_{++}^m$, and $p(.)$ is stable, $p(at+b)$ is real rooted. Let $\lambda_1,\dots,\lambda_m$ be the roots of this polynomial. We have
%$$ p(a+tb)=t^k p(a) \prod_{i=1}^m (1/t - \lambda_i)=p(a)\prod_{i=1}^m (1-t\lambda_i).$$
%So,
%$$ \log p(a+tb)=\log p(a) + \sum_{i=1}^m \log (1-t\lambda_i).$$
%The theorem follows by the fact that $\log(1-t\lambda_i)$ is a concave function for any $\lambda_i\in \R$.
%\end{proof}

The following theorem is proved by Gurvits \cite{Gur06}.
\begin{theorem}[\cite{Gur06}]
\label{thm:gurvits}
For any degree $m$-homogeneous stable polynomial $p(z_1,\dots,z_m)$ with positive coefficients, let $c_{[m]}$ denote the coefficient of the multilinear monomial $z_{1}\cdots z_{m}$. If $c_{[m]}>0$, then
$$ c_{[m]} \geq \frac{m!}{m^m} \inf_{\bz>0} \frac{p(z_1,\dots,z_m)}{z_1\dots z_m}.$$	
\end{theorem}

\subsection{Counting Matchings in Bipartite Graphs}
Given a   bipartite graph $G=(X,Y,E)$ with weights $w:E\to \R$.
The weight of a perfect matching $M$ is defined as follows:
$$ w(M)=\prod_{e\in M} w_e.$$
Jerrum, Sinclair, and Vigoda in their seminal work designed a FPRAS to count the sum of (weighted) perfect matchings of an arbitrary bipartite graph with nonnegative weights. This problem is also equivalent to the computation of the permanent of a non-negative matrix.
\begin{theorem}[\cite{JSV04}]
There exists a randomized polynomial time algorithm that for any arbitrary bipartite graph $G$ with $n$ vertices and nonnegative weights and $\eps>0$ in time polynomial in the size of $G$ and $1/\eps$ approximates the sum of weights of all perfect matchings of $G$ within a $1+\eps$ multiplicative error, with high probability.
\end{theorem}

A $k$-matching of a bipartite graph $G=(X,Y,E)$ is a set $M\subseteq E$ of size $|M|=k$ such that no two edges share an endpoint.
The following corollary follows immediately  from the above theorem. For completeness, we prove it in the appendix.
\begin{corollary}\label{cor:countkmatchings}
There is a randomized polynomial time algorithm that for any arbitrary bipartite graph $G$ with nonnegative edge weights and for any given $\eps>0$ and integer $k\leq n$ in time polynomial in the size of $G$ and $1/\eps$ approximates the sum of the weights of all $k$-matchings of $G$ within $1+\eps$ multiplicative error, with high probability.
\end{corollary}

%One of the most interesting classes of real stable polynomials is the class of determinant polynomials as observed by Borcea and Br\"and\'en \cite{BB08}.

\section{Approximation Algorithm for NSW Maximization}

In this section, we give an approximation algorithm for the NSW maximization problem. We begin by giving a mathematical programming relaxation that can be efficiently solved. For convenience, we will aim to optimize
$$\left(\prod_{i=1}^n v_i(\bx)\right)$$
which is the $n^{th}$ power of the NSW objective. Thus, it is enough to give an $e^{-n}$-approximation to the above objective. With a slight abuse of notation, we will also refer to problem of maximizing the new objective as the Nash welfare problem.
In \autoref{sec:algorithm}, we give a rounding algorithm that proves the guarantee claimed in Theorem~\ref{thm:nashwelfaremain}.

\subsection{Mathematical Programming Relaxation}\label{sec:program}

We use the following mathematical program.
\begin{equation}
\begin{aligned}
	\max_{\bx}  ~\inf_{\by>0: y^S\geq 1,\forall S\in{[m]\choose n}} ~~ &   \prod_{i=1}^n \left(\sum_{j=1}^m x_{i,j}v_{i,j}y_j\right),&\\
	\st ~~~~~~ &  \sum_{i=1}^n x_{i,j} \leq 1 & \forall\; 1\leq j\leq m,\\
%	&  \prod_{j\in S} y_j \geq 1 & \forall S\in{m\choose n}\\
	&  x_{i,j}\geq 0 & \forall i,j.
\end{aligned}
\label{cp:main}	
\end{equation}
%In this section we prove \autoref{thm:nashwelfaremain} using the main technical theorem \ref{thm:gurvitsextension}.

First, we show that \eqref{cp:main} is a relaxation of the Nash welfare problem and can be optimized in polynomial time to an arbitrary accuracy.
\begin{lemma}
The mathematical program \eqref{cp:main} is a relaxation of the Nash welfare problem and can be optimized in polynomial time.
\end{lemma}
\begin{proof}
Let $x^*\in\{0,1\}^{n\times m}$ be an optimal solution of the Nash welfare problem and let $\sigma: [m]\rightarrow [n]$ denote the assignment, i.e., $\sigma(j)=i$ if and only if $x^*_{ij}=1$.
We show that $x^*$ is a feasible solution \eqref{cp:main} of objective $\prod_{i=1}^n v_i(x^*)$.
Consider any $\by>0$ such that $y^S\geq 1$ for each $S\subseteq \binom{[m]}{n}$. Moreover let $\SS=\{S\in \binom{[m]}{n}: \forall i\in [n], \exists j\in S \textrm{ such that } x^*_{ij}=1\}$.  We have
\begin{eqnarray*}
\prod_{i=1}^n \left( \sum_{j=1}^m x^*_{i,j} v_{i,j} y_j\right)& = \sum_{S\in \SS} y^S \prod_{j\in S} v_{\sigma(j),j}\\
&\geq \sum_{S\in \SS} \prod_{j\in S} v_{\sigma(j),j}\\
&=\prod_{i=1}^n \left( \sum_{j=1}^m x^*_{i,j} v_{i,j} \right)
\end{eqnarray*}
as required where we use the fact that $y^S\geq 1$ for each $S\in \SS$. To show that the objective of the mathematical program equals $\prod_{i=1}^n v_i(x^*)$, we consider the solution $y^*_j=1$ for each $j\in [m]$.

To solve the mathematical program, we observe that the function $\ln{ \prod_{i=1}^n \sum_{j=1}^m x_{i,j}v_{i,j}y_j}$ is concave in $x$ and convex in $\ln \by$, where $\ln \by$ is the vector defined by taking logarithms of the vector $\by$ coordinate-wise. These follow from \autoref{thm:guler} and \autoref{fact:logconvex}. Moreover, the constraints on $\bx$ and $\ln \by$ are linear. Thus the above program can be formulated as a convex program and solved to an arbitrary accuracy.
\end{proof}

\subsection{Randomized Algorithm I}\label{sec:algorithm}

We now give a rounding algorithm that proves the required guarantee. \autoref{alg:rounding} will only satisfy the guarantee in expectation. Later, we show how to give a randomized algorithm that gives essentially the same guarantee with high probability.

\begin{algorithm}[h]
\begin{algorithmic}
\State Check whether the optimal solution has weight strictly more than zero using the bipartite matching algorithm. Return zero if answer is false.
\State Find an optimal solution $\bx^*$ to the mathematical program~\eqref{cp:main}.
\State Independently for each item $j\in [m]$, assign item $j$ to one agent where agent $i\in[n]$ is chosen with probability $x^*_{ij}$.
\end{algorithmic}
\caption{An Algorithm for NSW Maximization}
\label{alg:rounding}
\end{algorithm}

%Let $x,y$ be a feasible solution of \eqref{cp:main}. We assume that the objective is strictly positive.
%We use simple randomized rounding algorithm to assign items each item to an agent.  Independently for each  item $1\leq j\leq m$, we assign $j$ to agent $i$ with probability $x_{i,j}$. Since, $\sum_{i=1}^n x_{i,j}=1$, this yields a valid assignment of each item to an agent. The following lemma shows the expected welfare of our randomized algorithm.

%W
The first step of the algorithm can be implemented by a bipartite matching problem. Indeed consider the bipartite graph with one side as agents and other as items. We have an edge $(i,j)$ for agent $i$ and item $j$ if $v_{ij}>0$. The optimal solution to the NSW maximization problem is strictly positive if and only if this bipartite graph has a matching that includes an edge at every agent. Thus, we can check in polynomial time whether the optimal solution has weight zero. For the remainder of the section, we assume that the optimal solution is strictly positive.

We say $\bx\in\R_+^{n\times m}$ is a fractional allocation vector if for each $j\in [m]$, $\sum_{i=1}^n x_{i,j}=1$.
 Given any fractional allocation $\bx$, consider the following polynomial in variables $y_1,\ldots, y_n$,
$$ p_{\bx}(y_1,\dots,y_n)=\prod_{i=1}^n \left(\sum_{j=1}^m x_{i,j}v_{i,j}y_j\right).$$
Observe that $p_{\bx}(\by)$ is a degree $n$-homogenous polynomial in $m$ variables for any $\bx$ or the identically $0$ polynomial.

\begin{lemma}\label{lem:sumsquarefreemon}
We have the following.

\begin{enumerate}
\item For $S\subseteq [m]$ of size $n$, let $c_S$ denote the coefficient of $y^S$ in $p_{\bx^*}(\by)$. Then expected value of \autoref{alg:rounding} equals
$$\sum_{S\in{[m]\choose n}} c_S. $$
\item The optimal value of the relaxation \eqref{cp:main} is
$$\inf_{\by:  y^S\geq 1, \forall S\in{[m]\choose n}} p_{\bx^*}(\by).$$
\end{enumerate}

\end{lemma}
\begin{proof}
%Let $v(i)$ be a random variable denoting the value that $i$ receives in the allocation returned by the rounding algorithm, and
Let $X_{i,j}$ be the random variable indicating that $j$ is assigned to $i$. Then, the value that $i$ receives is $\sum_{j=1}^m X_{i,j} v_{i,j}.$
So, the expected value of the algorithm is
\begin{eqnarray*}
\E{\prod_{i=1}^n \sum_{j=1}^m X_{i,j}v_{i,j}} = \sum_{\sigma:[n]\to[m]}  \E{\prod_{i=1}^n X_{i,\sigma(i)}v_{i,\sigma(i)}} = \sum_{\sigma:[n]\to[m]} \P{\forall i: X_{i,\sigma(i)}=1}\prod_{i=1}^n v_{i,\sigma(i)}.
\end{eqnarray*}
where $\sigma$ is summed over all functions from $[n]$ to $[m]$.
Observe that $\P{\forall i: X_{i,\sigma(i)=1}}\neq 0$ only if $\sigma$ is a one-to-one function. In such a case, we have $\P{\forall i: X_{i,\sigma(i)}=1}=\prod_{i=1}^n x_{i,\sigma(i)}$ where we use the fact that each item is assigned independently. Therefore,
\begin{eqnarray*}
\E{\prod_{i=1}^n \sum_{j=1}^m X_{i,j}v_{i,j}} = \sum_{\sigma:[n]\underset{\text{one-to-one}}{\rightarrow}[m]} \prod_{i=1}^n x_{i,\sigma(i)}v_{i,\sigma(i)}.
\end{eqnarray*}
The lemma follows by the fact that for any one-to-one $\sigma$, the term $\prod_{i=1}^n x_{i,\sigma(i)}v_{i,\sigma(i)}$ on the RHS appears in the coefficient of the (square-free) monomial $\prod_{i=1}^n y_{\sigma(i)}$ of $p_{\bx^*}(\by)$. For any $S\in {[m]\choose n}$ the coefficient of $y^S$ in $p_{\bx^*}(\by)$ is the sum of all such terms where $\sigma([n])=S$.

The proof of the second claim is immediate by definition.
\end{proof}

We are now ready to apply \autoref{thm:gurvitsextension} and obtain the following immediate corollary.
\begin{corollary}\label{cor:expectedwelfare}
 The expected objective of \autoref{alg:rounding} is at least
$$ e^{-n} \cdot \OPT$$
where $\OPT$ is the optimal NSW objective.
\end{corollary}
\begin{proof}
From \autoref{fact1} and \autoref{fact2}, it follows that $p_{\bx^*}(\by)$ as defined above is real stable with non-negative coefficients.  Moreover, it is an $m$-variate polynomial that is degree $n$-homogenous. Let $c_S$ denote the coefficient of square-free monomial $y^S$ for any $S\in \binom{[m]}{n}$. Since, we assume that there is at least one assignment that has strictly positive NSW objective, the sum of coefficients $\sum_{S\in \binom{[m]}{n}}c_S>0$. Thus, from \autoref{thm:gurvitsextension}, we have
\begin{eqnarray} \sum_{S\in{[m]\choose n}} c_S &\geq & e^{-n} \min_{\by:  y^S\geq 1, \forall S\in{[m]\choose n}} p_{\bx^*}(\by).	\nonumber
\end{eqnarray}

Now the proof is immediate using \autoref{lem:sumsquarefreemon}.
\end{proof}

\subsection{Randomized Algorithm II}
From \autoref{cor:expectedwelfare}, the expected NSW of the allocation returned by \autoref{alg:rounding} is at least $1/e^n$ fraction of the optimum.
Repeated applications of the algorithm to obtain a high probability bound is not possible since the output of \autoref{alg:rounding} may have an exponentially large variance.
 In this section, we prove Theorem~\ref{thm:nashwelfaremain} by giving an algorithm that returns the same guarantee as \autoref{alg:rounding} with high probability.
\begin{proofof}{\autoref{thm:nashwelfaremain}}
We use the method of conditional expectations to prove the theorem. We iteratively assign one item at a time, making sure that conditional expectation over the random assignment of the remaining items does not decrease (substantially). We now claim that
for any assignment $x$, the expected value of the objective as given by randomized algorithm \autoref{alg:rounding} equals the number of weighted $n$-matchings of a bipartite graph. Consider the weighted bipartite graph $G=([n],[m],E)$ where for any $1\leq i\leq n$ and $1\leq j\leq m$, $w_{i,j}=x_{i,j}v_{i,j}$. Then, for one-to-one mapping $\sigma:[n]\to[m]$, the coefficient of the monomial $\prod_{i=1}^n x_{i,\sigma(i)}v_{i,\sigma(i)}$ is equal to the weight of the $n$-matching $\{(1,\sigma(1)),(2,\sigma(2)),\dots,(n,\sigma(n))\}$. Therefore, the sum of square-free monomials of $p_{\bx}(\by)$ is equal to the sum of the weights of all $n$-matchings of $G$.

Now, pick any item $j\in [m]$ and any fractional assignment $x$. Consider the following $n$ assignments, $x^1,\ldots, x^n$. Assignment $x^{i}$ assigns item $j$ to $i$ and rest of the items identically to the fractional assignment $x$. Thus $x^i_{ij}=1$, $x^i_{i',j}=0$ for all $i\neq i'$ and $x^{i}_{i'j'}=x_{i'j'}$ for any $j'\neq j$. Let $ALG^i$ denote the objective value of the output of \autoref{alg:rounding} on solution $x^i$ and $ALG$ on $x$. Since the objective value of the \autoref{alg:rounding} is linear in $\{x_{ij}:{i\in [n]}\}$ for fixed $j$, we have
\begin{eqnarray*}
ALG=\sum_{i=1}^n x_{ij} ALG^i
\end{eqnarray*}

Thus $ALG$ is the expected value of the conditional expected value of the output of the \autoref{alg:rounding} when we assign item $j$ to one of the agents; it is assigned to agent $i$ with probability $x_{ij}$.

By \autoref{cor:countkmatchings}, we can estimate $ALG$ and $ALG^i$ within a factor of $1+1/m^3$ factor in polynomial time. Therefore, using the method of conditional expectations, we obtain an allocation of NSW of value at least $\frac{\OPT}{e^n}\cdot (1-1/m^3)^m\geq \frac{\OPT}{\left((1+\frac{1}{n})e\right)^n}$ where $\OPT$ denotes the objective of the optimal allocation.
\end{proofof}

\section{A Generalization of Gurvits's Theorem}
In this section we prove \autoref{thm:gurvitsextension}.
Let $$q(y_1,\dots,y_m)=(y_1+\dots+y_m)^{m-n} $$ be a degree $(m-n)$-homogenous polynomial. It is straightforward to see that it is real stable.
Consider the polynomial $p(\by)q(\by)$. Observe that this is a degree $m$-homogeneous stable polynomial with non-negative coefficients. Since from the assumption of \autoref{thm:gurvitsextension}, at least one of the square-free monomials in $p(\by)$ has a non-zero coefficient, the coefficient of the square-free monomial in  $p(\by)q(\by)$ is non-zero. Let $\alpha_{[m]}$ be the coefficient of the square-free monomial $y_1\cdots y_m$ in $p(\by)q(\by)$.  Thus, from \autoref{thm:gurvits}, we have
\begin{equation}\label{eqn:pqgurivts}
\alpha_{[m]} \geq \frac{m!}{m^m} \inf_{\by>0} \frac{p(\by)q(\by)}{y_1\dots y_m}.
\end{equation}

To prove \autoref{thm:gurvitsextension} it is enough to relate the LHS and the RHS of \eqref{eqn:pqgurivts} to the two sides of \eqref{eq:gurvitsextension}.
This is done in \autoref{lem:ptopq} and \autoref{prop:pqtop}.
\begin{lemma}\label{lem:ptopq}We have
	$$ (m-n)!\sum_{S\in{[m]\choose n}} c_S = \alpha_{[m]}. $$
	%$$ (m-k)!\left(\sum_{S\in{m\choose k}} \prod_{i\in S} \partial_{y_i}\right) p(y) \Bigg|_{y_1=\dots=y_m=0}= \left(\prod_{i=1}^m \partial_{y_i}\right) p(y)q(y)\Big|_{y_1=\dots=y_m=0}.$$
\end{lemma}
\begin{proof}
The RHS is the coefficient of the square-free monomial $y_1\dots,y_m$ in $p(\by)q(\by)$.
The square-free monomial of $p(\by)q(\by)$ is obtained whenever we multiply a square-free monomial $y^S$ of $p(\by)$ with the square-free monomial $y^{\overline{S}}$ of $q(\by)$ for some $S\in \binom{[m]}{n}$. Lemma's statement follows by the fact that the coefficient of $y^{\overline{S}}$ in $q(\by)$ is $(m-n)!$ for every $S\in \binom{[m]}{n}$ and the coefficient of $y^S$ in $p(\by)$ is $c_S$.
\end{proof}

The proof of \autoref{thm:gurvitsextension} is now immediate from the following proposition which relates the RHS of \eqref{eqn:pqgurivts} and \eqref{eq:gurvitsextension}.
\begin{proposition}\label{prop:pqtop}
	$$\inf_{\by>0} \frac{p(\by)q(\by)}{y_1\dots y_m} \geq (m-n)^{m-n}\inf_{\by>0: y^S\geq 1,\forall S\in{m\choose n}} p(\by).$$
\end{proposition}

In the rest of this section we prove the above proposition. We do the proof in two steps. First, we use convex duality to simplify the RHS, and then we prove the proposition.

\begin{lemma}\label{lem:duality}
$$ \inf_{\by>0: y^S\geq 1,\forall S\in{m\choose n}} p(\by) = \sup_{0\leq \theta\leq 1: \sum_{j=1}^m \theta_j=n} \inf_{\by>0} \frac{p(\by)}{y_1^{\theta_1}\dots y_m^{\theta_m}}.$$	
\end{lemma}
\begin{proof}
	The proof follows by convex duality. By taking logarithm of $p(\by)$ and the change of variable $z_j=e^{y_j}$, we obtain the following equivalent convex program to the LHS of the above inequality.
	\begin{equation}
	\begin{aligned}
		\inf ~~& \log p(e^{z_1},\dots,e^{z_m}) & \\
		\st ~& \sum_{i\in S} z_i\geq 0 & \forall S\in{[m]\choose n}.
	\end{aligned}
	\label{cp:y}
	\end{equation}
	Let $\lambda_S$ be the Lagrange dual variable associated to the constraint corresponding to the set $S\in \binom{[m]}{n}$. The Lagrangian of the above convex program is defined as follows:
	$$ L(\bz,\blambda )=\log p(e^{z_1},\dots,e^{z_m}) - \sum_{S\in\binom{[m]}{n}} \lambda_S \sum_{i\in S} z_i.$$	
	The Lagrange dual to \eqref{cp:y} is
	$$ \sup_{\blambda \geq 0} \inf_{\bz} L(\bz,\blambda).$$
	Since $p(\by)$ has a non-zero coefficient for at least one of the square-free monomials, the objective of the convex program \eqref{cp:y} is finite for any $\bz$ and it is easy to see that Slater conditions are satisfied. Thus the optimum value of the Lagrange dual is exactly equal to the optimum of \eqref{cp:y}.
	
	Let $\bz^*,\blambda^*$ be an optimum of the above program. We claim that $\sum_S \lambda^*_S=1.$ This simply follows from first order optimality conditions. If $\sum_S \lambda^*_S<1$, then
	\begin{eqnarray*} L(\bz^*-\eps,\blambda^*) &=& \log p(e^{z^*_1-\eps},\dots,e^{z^*_m-\eps}) - \sum_{S\in {[m]\choose n}} \lambda^*_S \sum_{j\in S} (z^*_j-\eps)\\
	 &=& L(\bz^*,\blambda^*) - n\cdot \eps + \sum_{S\in {[m]\choose n}} n\lambda^*_S \eps < L(\bz^*,\blambda^*). 	
	\end{eqnarray*}
	Similarly, if $\sum_{S\in{[m]\choose n}} \lambda_S >1$, $L(\bz^*+\eps,\blambda^*) < L(\bz^*,\blambda^*)$.
	So, $\blambda^*$ is a probability distribution on sets of size $n$. We let $L'(\bz,\theta)=p(e^{\bz})-\sum_{j=1}^m z_j\theta_j$
Thus, we obtain that
	$$ \sup_{0\leq \theta\leq 1: \sum_{j=1}^m \theta_j=n} \inf_{\bz} L'(\bz,\theta) \geq  \sup_{\blambda \geq 0} \inf_{\bz} L(\bz,\blambda).$$
by setting $\theta^*_j=\sum_{S\in {[m]\choose n}: j\in S} \lambda^*_S$ to be the marginal probability of the element $j$.

We now claim that equality must hold in the above. This follows since given any $\theta \in \{0\leq \theta\leq 1: \sum_{j=1}^m  \theta_j=n\}$, there exists a probability distribution over sets of size $n$ such that marginal of every element is exactly $\theta_j$. Setting $\lambda'_S$ to be the probability of set $S\in \binom{[m]}{n}$, we obtain that for any $\bz$ and $\theta$, we have $L'(\bz,\theta)=L(\bz,\blambda')$.
	Putting this together we have,
	$$ \inf_{\sum_{j\in S} z_j\geq 0, \forall S\in{[m]\choose n}}  \log p(e^\bz) = \sup_{0\leq \theta\leq 1: \sum_{j=1}^m \theta_j=n} \inf_{\bz} \log p(e^\bz)-\sum_{j=1}^m z_j\theta_j.$$
	Substituting $e^{z_j}$ to $y_j$ and taking the exponential of the objective functions we have
	$$ \inf_{\by >0 :y^S\geq 1,\forall S\in{[m]\choose n}} p(y) = \sup_{0\leq \theta\leq 1: \sum_{j=1}^m \theta_j=n} \inf_{\by>0} \frac{p(y)}{y_1^{\theta_1} \dots y_m^{\theta_m}}$$
	as desired.
\end{proof}

Now we give the proof of \autoref{prop:pqtop}.
\begin{proofof}{\autoref{prop:pqtop}}
By \autoref{lem:duality}, it is enough to show that
$$\inf_{\by>0} \frac{p(\by)q(\by)}{y_1\dots y_m}\geq (m-n)^{m-n}\sup_{0\leq \theta\leq 1: \sum_{j=1}^m \theta_j=n} \inf_{\by>0} \frac{p(\by)}{y_1^{\theta_1}\dots y_m^{\theta_m}}.$$
Let $\theta$ be any vector such that $0\leq \theta\leq 1$ and $\sum_i \theta_i=k$. Equivalently, it is enough to show for any such $\theta$,
$$ \inf_{\by>0} \frac{p(\by)q(\by)}{y_1\dots y_m}\geq  (m-n)^{m-n}\inf_{\by>0} \frac{p(\by)}{y_1^{\theta_1}\dots y_m^{\theta_m}}.$$

We prove a stronger statement,
$$ \inf_{\by>0} \frac{p(\by)}{y_1^{\theta_1}\dots y_m^{\theta_m}}\cdot \inf_{\by>0} \frac{q(\by)}{y_1^{1-\theta_1}\cdot y_m^{1-\theta_m}}\geq  (m-n)^{m-n}\inf_{\by>0} \frac{p(\by)}{y_1^{\theta_1}\dots y_m^{\theta_m}}.$$
Equivalently, we show that
$$	\inf_{\by>0} \frac{q(\by)}{y_1^{1-\theta_1}\dots y_m^{1-\theta_m}}\geq  (m-n)^{m-n}$$
Taking $(m-n)$-th root of both sides it is enough to show that
\begin{equation}\label{eq:homyalpha} \inf_{\by>0} \frac{y_1+\dots+y_m}{y_1^{\alpha_1}\dots y_m^{\alpha_m}}\geq m-n,	
\end{equation}
where $\alpha_j=\frac{1-\theta_j}{m-n}$ for all $j\in [m]$. Note that by definition of $\theta$, we have $0\leq \alpha_j\leq \frac{1}{m-n}$ and that
$$\sum_i \alpha_j=\frac{m-\sum_{j=1}^m \theta_j}{m-n}=1.$$
Therefore, the ratio in the LHS of \eqref{eq:homyalpha} is homogeneous in $\by$. Thus, to prove \eqref{eq:homyalpha}, it is enough to prove the following
\begin{equation}\label{eq:ysumone} \sup_{\by>0: \sum_{j=1}^m y_j=1} y_1^{\alpha_1} \dots y_m^{\alpha_m} \leq \frac{1}{m-n}.	
\end{equation}
Next, we use the weighted AM-GM inequality.
We let $\alpha_1,\dots,\alpha_m$ be the weights, and recall that $\alpha_j$'s sum to 1. Weighted AM-GM implies that
$$ \sum_{j=1}^m \alpha_j \frac{y_j}{\alpha_j} \geq \prod_{j=1}^m \left(\frac{y_j}{\alpha_j}\right)^{\alpha_j} = \prod_{j=1}^m \alpha_j^{-\alpha_j} \prod_{j=1}^m y_j^{\alpha_j}$$
Therefore,
$$ \sup_{\by>0: \sum_{j=1}^m y_j =1} \prod_{j=1}^m y_j^{\alpha_j} \leq \prod_{j=1}^m \alpha_j^{\alpha_j}.$$
To prove \eqref{eq:ysumone}, it is enough to show that
$$ \prod_{j=1}^m \alpha_j^{\alpha_j} \leq \frac{1}{m-n}.$$
Or equivalently,
$$ \sum_{j=1}^m -\alpha_j \log \alpha_j \geq \log(m-n).$$
Since $\alpha_j\leq \frac{1}{m-n}$ and that
$\sum_{j=1}^m \alpha_j=1$, we have

\begin{eqnarray*}
 \sum_{j=1}^m -\alpha_j \log \alpha_j &\geq &\sum_{j=1}^m -\alpha_j \log{\frac{1}{m-n}} =   \log{(m-n)} \sum_{j=1}^m \alpha_j= \log{(m-n)},
\end{eqnarray*}
as required.
\end{proofof}

\bibliographystyle{alpha}

\bibliography{ref2}

\appendix
\section{Miscellaneous Lemmas}
\begin{lemma}
\label{lem:etomk}
For any $k\leq m$, we have
$$ \frac{m!}{m^m}\cdot \frac{(m-k)^{m-k}}{(m-k)!} \geq e^{-k}.$$	
\end{lemma}
\begin{proof}
We prove by induction on $k$. The claim obviously holds for $k=0$. For the induction step, it is sufficient to show that
$$ \frac{1}{e}\cdot \frac{m!}{m^m} \cdot \frac{(m-k)^{m-k}}{(m-k)!} \leq \frac{m!}{m^m} \cdot\frac{(m-(k+1))^{m-(k+1)}}{(m-(k+1))!}.$$	
Equivalently, it is enough to show that
$$ \left(\frac{m-k}{m-(k+1)}\right)^{m-(k+1)}\leq e.$$
The above can be written as $(1+\frac{1}{m-k-1})^{m-k-1} \leq e$. The latter follows by the fact that $1+x \leq e^x$.
\end{proof}

\begin{proofof}{Corollary~\ref{cor:countkmatchings}}
Suppose that we are given a bipartite graph $G=(X,Y,E)$ where $X=\{x_1,\dots,x_m\}$ and $Y=\{y_1,\dots,y_n\}$. Note that $m$ is not necessarily equal to $n$. We construct another graph $G'=(X',Y',E')$ such there is a one-to-$(m-k)!(n-k)!$ and onto mapping between the $k$-matchings of $G$ and the perfect matchings of $G'$.
That is each $k$-matching of $G$ is mapped to a unique set of $(m-k)!(n-k)!$ perfect matchings of $G'$, and for each perfect matching $M'$ of $G'$ there is a $k$-matching of $G$ that has $M'$ in its image.

 Let $X'=X\cup\{x_{m+1},\dots,x_{m+n-k}\}$ and $Y'=Y\cup\{y_{n+1},\dots,y_{m+n-k}\}$.
The set of edges $E'$ is the union of $E$ and the following edges: Connect all vertices of $X'\setminus X$ to all vertices of $Y$ with weight 1, and connect all vertices of $Y'\setminus Y$ to all vertices of $X$ with weight 1. Observe that for any $k$-matching $M$ of $G$ there are exactly $(m-k)!(n-k)!$ perfect matchings in $G'$ that contain $M$; for any such prefect matching $M'$, we have $M'\setminus M\subseteq E'\setminus E$. So, this mapping is one-to-$(m-k)!(n-k)!$. Furthermore, any perfect matching $M'$ of $G'$ has exactly $k$ edges in $E$, i.e., $|M'\cap E|=k$. So, this mapping is onto.

%we have $w(M')=w(M)$. This

It follows that a $1+\eps$  approximation to the sum of the weights of all perfect matchings of $G'$ is a $1+\eps$ approximation to the sum of the weights of all $k$-matchings of~$G$.
\end{proofof}

%\printbibliography

\end{document}